\documentclass{article}
\usepackage{color}
\usepackage{algorithm}
\usepackage[noend]{algorithmic}
\usepackage{amsfonts}
\usepackage{amsthm}
\usepackage{times}
\usepackage{amsmath}
\usepackage{fullpage}

\title{Greedy Minimization of Weakly Supermodular Set Functions}

\author{Christos Boutsidis\thanks{boutsidis@yahoo-inc.com, Yahoo Labs, New York, NY}
\and
Edo Liberty\thanks{edo@yahoo-inc.com, Yahoo Labs, New York, NY}
\and
Maxim Sviridenko\thanks{sviri@yahoo-inc.com, Yahoo Labs, New York, NY}
}
\date{}

\newcommand{\eps}{\varepsilon}
\newcommand{\R}{\mathbb{R}}

\newcommand{\T}{{\scriptscriptstyle{  T}}}
\newcommand{\inv}{{\scriptstyle{  \textrm +}}}

\newtheorem{definition}{Definition}
\newtheorem{theorem}{Theorem}
\newtheorem{lemma}{Lemma}
\begin{document} 
\maketitle

\begin{abstract}
This paper defines {\it weak-$\alpha$-supermodularity} for set functions.
Many optimization objectives in machine learning and data mining seek to minimize such functions under cardinality constrains.
We prove that such problems benefit from a greedy extension phase.
Explicitly, let $S^*$ be the optimal set of cardinality $k$ that minimizes $f$ and let $S_0$ be an initial solution such that $f(S_0)/f(S^*) \le \rho$.
Then, a greedy extension $S \supset S_0$ of size $|S| \le |S_0| + \lceil \alpha k \ln(\rho/\eps) \rceil$ yields $f(S)/f(S^*) \le 1+\eps$.
As example usages of this framework we give new bicriteria results for $k$-means, sparse regression, and columns subset selection.
\end{abstract}

\section{Introduction}
Many problems in data mining and unsupervised machine learning take the form of minimizing set functions with cardinality constraints.
More explicitly, denote by $[n]$ the set $\{1,\dots,n\}$ and $f(S):2^{[n]}\rightarrow \R_+$.
Our goal is to minimize $f(S)$ subject to $|S|\le k$. 
These problems include clustering and covering problems as well as regression, matrix approximation problems and many others.
These combinatorial problems are hard to minimize in general.
Finding good (e.g.\ constant factor) approximate solutions for them requires significant sophistication and highly specialized algorithms.

In this paper we analyze the behavior of the greedy algorithm to all these problems.
We start by claiming that the functions above are special.
A trivial observation is that they are non-negative and non-increasing, that is, $f(S \cup T)  \le  f(S)$ for any $S,T \subseteq [n]$.
This immediately shows that expanding solution sets is (at least potentially) beneficial in terms of reducing the function value.
But, monotonicity is not enough to ensure that any number of greedy extensions of a given solution would reduce the objective function.

To this end we need to somehow quantify the gain of adding a single element (greedily) to a solution set.
Let $f(S)- f(S\cup T)$ be the reduction in $f$ one gains by adding a set of elements $T$ to the current solution $S$.
Then, the average gain of adding elements from $T$ \emph{sequentially} is $[f(S)- f(S\cup T)]/|T \setminus S|$.
One would hope that there exists an element in $i \in T \setminus S$ such $f(S)- f(S\cup \{i\}) \ge [f(S)- f(S\cup T)]/|T \setminus S|$
but that would be false, in general, since the different element contributions are not independent of each other.
Lemma~\ref{name}, however, shows that this is true for supermodular functions. 
\begin{definition}
A  set function $f(S):2^{[n]}\rightarrow \R_+$ is said to be {\it supermodular} 
if  for any two sets $S,T\subseteq [n]$
\begin{equation}\label{supermo1}
f(S\cap T)+ f(S\cup T) \ge  f(S )+f( T).
\end{equation}
\end{definition}

Combining this fact with the idea that $T$ could be any set, 
including the optimal solution $S^*$, already gives some useful results for minimizing supermodular set functions.
Specifically those for which $f(S^*)$ is bounded away from zero.
Notice that $k$-means is exactly this kind of problem. 
Section~\ref{clustering} gives some new bicriteria results obtainable for $k$-means via the greedy extension algorithm of Section~\ref{main}.
A similar intuition gives a very famous result that the greedy algorithm provides a $(1- 1/e)$-factor approximation for maximizing set functions $g(S)$ subject to $|S|\le k$ if $g$ for positive, monotone non-decreasing and submodular \cite{NemhauserWF1978}.

Alas, most problems of interest, such as regression, columns subset selection, feature selection, and outlier detection (and many others) are not supermodular.
In Section~\ref{sviri-rocks} we define the notion of weak-$\alpha$-supermodularity. 
Intuitively, weak-$\alpha$-supermodular functions are those conducive to greedy type algorithms.
Or, alternatively, the inequality above holds up to some constant $\alpha > 1$.
Weak-$\alpha$-supermodularity requirers that there exists an element $i \in T \setminus S$ such that adding $i$ \emph{first} 
gains at least $[f(S)- f(S\cup T)]/\alpha|T \setminus S|$ for some $\alpha \ge 1$.

As an example for this framework we show in Section~\ref{smlr} that Sparse Multiple Linear Regression (SMLR) is weak-$\alpha$-supermodular.
Using this fact we extend (and slightly improve) the result of \cite{Natarajan1995} for Sparse Regression and 
obtain new bicriteria results for Columns Subset Selection.

\section{Weakly Supermodular Set Functions}\label{sviri-rocks}
In this section, we define our notation and the notion of {\emph weak-$\alpha$-supermodularity}.
Throughout the manuscript we denote by $[n]$ the set $\{1,\dots,n\}$.
We concern ourselves with non-negative set function $f(S):2^{[n]}\rightarrow \R_+$.
More specifically monotone non-increasing set function such that $f(S) \ge f(S \cup T)$ for any two sets $S \subseteq [n]$ and $T\subseteq[n]$.

\begin{definition}
A non-negative non-increasing set function $f(S):2^{[n]}\rightarrow \R_+$ is said to be {\it weakly-$\alpha$-supermodular} 
if there exists $\alpha \ge 1$ such that for any two sets $S,T\subseteq [n]$
\begin{equation}\label{supermo}
f(S)- f(S\cup T) \le \alpha\cdot|T\setminus S|\cdot \max_{i \in T\setminus S}[f(S)- f(S\cup \{i\})] \ .
\end{equation}
\end{definition}
This property is useful because we will later try to minimize $f$.   
It asserts that if adding $T \setminus S$ is beneficial then there is an element $i \in T\setminus S$ that contributes at least a fraction of that.
The reason for the name of this property might also be explained by the following definition and  lemma.

\begin{lemma}\label{name}
A non-increasing non-negative supermodular function $f$ is weakly-$\alpha$-supermodular with parameter $\alpha=1$. 
\end{lemma}
\begin{proof} 
For $S,T \subseteq [n]$ order the set $T\setminus S$ in an arbitrary order, i.e. $T\setminus S=\{i_1,\dots,i_{|T\setminus S|}\}$.
Define $R_0=\emptyset$ and  $R_t=\{i_1,\dots i_t\}$ for $t>0$.
By supermodularity we have for any $t$
\begin{eqnarray}\label{superaa}
f(S)- f(S\cup \{i_{t}\}) \ge f(S\cup R_{t-1})- f(S\cup R_{t-1} \cup \{i_{t}\} )
\end{eqnarray}
We note that $R_{t-1} \cup \{i_{t}\} = R_{t}$  and sum up Equation~\eqref{superaa}.
\begin{eqnarray*}
\sum_{t=1}^{|T\setminus S|}  [f(S)- f(S\cup \{i_t\})] \ge \sum_{t=1}^{|T\setminus S|} f(S\cup R_{t-1})- f(S\cup R_{t-1}\cup \{i_{t}\}) = f(S)- f(S\cup T) \ .
\end{eqnarray*}
Since  $|T\setminus S|\cdot\max_{i \in T\setminus S}[f(S)- f(S\cup \{i\})] \ge \sum_{t=1}^{|T\setminus S|}  [f(S)- f(S\cup \{i_t\})]$ this implies weak-$1$-supermodularily.
\end{proof}

\section{Greedy Extension Algorithm}\label{main}
We are given a non-increasing weakly-$\alpha$-supermodular set function $f(S)$ and would like to solve the following optimization problem
\begin{equation}\label{mainProblem}
\min \{f(S):  |S|\le k\}.
\end{equation}

Consider a simple greedy algorithm that starts with some initial solution $S_0$ of 
value $f(S_0)$ (maybe $S_0=\emptyset$) and sequentially and greedily adds elements to it to minimize $f$.  
\begin{algorithm}[h!]
\begin{algorithmic}
\STATE {\bf input:} Weakly-$\alpha$-supermodular function $f$,  $S_0$, $k$, $E$
\FOR {$t = 1,\ldots, \lceil \alpha k \ln( f(S_0)/E)\rceil$}
	\STATE $S_{t} \gets S_{t-1} \cup \arg\min_{i \in [n]} f(S_{t} \cup \{i\})$
\ENDFOR
\STATE {\bf output:} $S_t$
\end{algorithmic}
\caption{Greedy Extension Algorithm}\label{greedy}
\end{algorithm}

 \begin{theorem}\label{mainSupermodular}
Let $S_\tau$ be the output of Algorithm~\ref{greedy}. Then $|S_\tau| \le |S_0| + \lceil \alpha k \ln( f(S_0)/E)\rceil$ 
and $f(S_\tau) \le f(S^*)+E$  where $S^*$ is an optimal solution of the optimization problem \eqref{mainProblem}.
\end{theorem}
\begin{proof}
The fact that $|S_\tau| \le |S_0| + \lceil \alpha k \ln( f(S_0)/E)\rceil$ is a trivial observation.
%
For the second claim consider an arbitrary iteration $t \in  [\tau]$ and consider   the set $S^*\setminus S_{t-1}$. 
By monotonicity and weak $\alpha$-supermodularity
\begin{eqnarray*}
f(S_{t-1})-f( S^* ) &\le&  f(S_{t-1})-f(S_{t-1}\cup S^* ) \le \alpha k \cdot \max_{i \in S^* \setminus S_{t-1}} f(S_{t-1})-f(S_{t-1} \cup \{i\})\\
&\le& \alpha k \cdot\max_{i \in [n]} f(S_{t-1})-f(S_{t-1} \cup \{i\}) = \alpha k \cdot ( f(S_{t-1})-f(S_{t})) \ .
\end{eqnarray*}
By rearranging the above equation and recursing over $t$ we get
\begin{equation*}
f(S_t)-f(S^*) \le \left( f(S_{t-1})-f(S^* ) \right) \left(1-1/\alpha k \right) \le\left( f(S_{0})-f(S^* ) \right)\left(1 - 1/\alpha k\right)^t 
\end{equation*}
\noindent Substituting $\tau = \lceil \alpha k \ln( f(S_0)/E)\rceil$ for the last step of the algorithm completes the proof.
\begin{eqnarray*}
f(S_\tau)-f(S^*) &\le&\left( f(S_{0})-f(S^* ) \right)\left(1-1/\alpha k\right)^{\alpha k \ln( f(S_0)/E)} \\
&\le&\left( f(S_{0})-f(S^* ) \right)e^{-\ln( f(S_0)/E)} \le E.
\end{eqnarray*}
\end{proof}
\begin{theorem}\label{multiplicative}
Assume there exist a $\rho$-approximation algorithm creating $S_0$ such that $f(S_0) \le \rho f(S^*)$.
There exists an algorithm for generating $S$ such that $|S| \le  |S_0| + \lceil \alpha k\left( \ln \frac{\rho }{\eps} \right)\rceil$ and $f(S) \le (1+\eps)f(S^*)$.
\end{theorem}
\begin{proof}
Use the $\rho$-approximation algorithm to create $S_0$ for Algorithm~\ref{greedy} and set $E = \eps f(S_0)$.
\end{proof}
\begin{algorithm}[h!]
\begin{algorithmic}
\STATE {\bf input:} Weakly-$\alpha$-supermodular function $f$,  $S_0$, $f_{\operatorname{stop}}$
\REPEAT  
	\STATE $S_{t} \gets S_{t-1} \cup \arg\min_{i} f(S_{t-1}\cup \{i\})$ 
\UNTIL {$f(S_{t}) \le f_{\operatorname{stop}}$}
\STATE {\bf output:} $S = S_t$
\end{algorithmic}
\caption{Greedy Extension Algorithm; an alternative stopping criterion}\label{greedy2}
\end{algorithm}

\begin{theorem}\label{modification}
Let $k_f$ be the minimal cardinality of a set $S'$ such that $f(S') \le f$.
For any $ f_{\operatorname{stop}}$ such that $f < f_{\operatorname{stop}}$  \ Algorithm~\ref{greedy2} outputs $S$ such that
$$|S|\le |S_0|+ \left\lceil \alpha k_f  \left(\ln \frac{f(S_0)}{f_{\operatorname{stop}}-f} \right) \right\rceil$$
\end{theorem}
\begin{proof}
The proof follows from Theorem \ref{mainSupermodular} by setting $k = k_f$ and $E = f_{\operatorname{stop}}-f$.
\end{proof}  

\section{Clustering}\label{clustering}
We will use the following auxiliary problem.
\begin{definition}[  $k$-Median]
We are given a set $X$ of data points, the set ${\cal C}$ of potential cluster center locations and the nonnegative costs $w_{ij}\ge 0$ for all $i,j\in X\times {\cal C}$. Find a set $S \subset {\cal C}$ minimizing $f(S)  = \sum_{i\in X} \min_{j \in {\cal C}} w_{ij}$ subject to $|S| \le k$.
\end{definition}
It is well known that the objective function $f(S)$ of the $k$-Median problem is supermodular and therefore weakly-$1$-supermodular by Lemma \ref{name}.
Our first application is a constrained version of the  $k$-means clustering problem.
\begin{definition}[Constrained $k$-Means]
Given a set of points $X \subset \R^d$, find a set $S \subset X$ minimizing $f(S)  = \sum_{x\in X} \min_{x' \in S} \|x- x'\|^2$ subject to $|S| \le k$.
\end{definition}
\begin{lemma}
Given a set of $n$ points $X$ define $S^*$ the optimal solution to the constrained $k$-means problem. 
Namely, $S^*$ minimizes $f(S)$ subject to $|S|\le k$.
One can find in $O(n^2dk\log(1/\eps))$ time a set $S$ of size $|S| =  O(k) + k \log (1/\eps)$ such that $f(S) \le (1+\eps)f(S^*)$.
\end{lemma}
\begin{proof}
The constrained $k$-means objective function $f$ is weakly-$1$-supermodular because the problem is a special case of the $k$-Median problem defined above.
Using the the algorithm of \cite{AggarwalDK09} one obtains a set $S_0$ of size $|S_0| = O(k)$ points from the data for which $f(S_0) = O( f(S^*))$.
Their technique improves on the analysis of adaptive sampling method of \cite{ArthurV07}.
Greedily extending $S_0$ and applying the analysis of Theorem~\ref{mainSupermodular} completes the proof.
The quadratic dependency of the running time on the number of data points can be alleviated using the corset construction of \cite{FeldmanFSS2007,FeldmanL2011}
\end{proof}
The classical   $k$-means clustering problem  is defined as follows.
\begin{definition}[Unconstrained $k$-Means]
Given a set of  $n$ points $X \subset \R^d$, find a set $S \subset \R^d$ minimizing $f(S)  = \sum_{x\in X} \min_{c \in S} \|x- c\|^2$ subject to $|S| \le k$.
\end{definition}
\begin{lemma}\label{fast_kmeans}
Let $f(S^*)$ be the optimal solution to the unconstrained $k$-means problem.
One can find in time $O(n^2dk\log(1/\eps))$ a set $S \in \R^d$ of size $|S| =  O(k) + k \log (1/\eps)$ such that $f(S) \le (2+\eps)f(S^*)$.
\end{lemma}
\begin{proof}
The proof and the algorithm are identical to the above. 
The only point to note is that a $1+\eps/2$ approximation to the constrained problem is at most a $2+\eps$ approximation to the unconstrained one.
See \cite{ArthurV07}, for example, for the argument that the minimum of the constrained objective is at most twice that of the unconstrained one.
\end{proof}
Alternatively, we can utilize a more computationally expensive approach.
It is known that given an instance $(X,k)$ of the Unconstrained $k$-Means problem one can construct in polynomial time an instance of the $k$-Median problem $(X,{\cal C},w,k)$ where ${\cal C}\subseteq   \R^d$ such that for any solution of value $\Phi$ for the Unconstrained $k$-Means problem  there exists a solution of value $(1+\varepsilon)\Phi$ for the corresponding instance of the  $k$-Median problem (see Theorem 7 \cite{MMSW}). Moreover, $|{\cal C}|= n^{O(\log(1/\varepsilon)/\varepsilon^2)}$. Therefore, after applying this transformation on our instance of the Unconstrained $k$-Means and using the same initial solution $S_0$ as in Lemma \ref{fast_kmeans} we derive.
\begin{lemma}\label{slow_kmeans}
Let $f(S^*)$ be the optimal solution to the unconstrained $k$-means problem.
One can find in time $O(n^{O(\log(1/\varepsilon)/\varepsilon^2)}dk)$ a set $S \in \R^d$ of size $|S| =  O(k) + k \log (1/\eps)$ such that $f(S) \le (1+\eps)f(S^*)$.
\end{lemma}

\section{Sparse Multiple Linear Regression}\label{smlr}

We begin by defining the Sparse Multiple Linear Regression (SMLR) problem.
Given two matrices $X \in \R^{m \times n}$ and $Y \in \R^{m \times \ell}$, and an integer $k$
find a matrix $W \in \R^{n \times \ell}$ that minimizes $\|XW - Y\|_F^2$ subject to $W$ having at most $k$ non zero rows. 
We assume for notational brevity (and w.l.o.g.) that the columns of $X$ have unit norm. 
An alternative and equivalent formulation of SMLR is as follows.
Let $X_S$ be a submatrix of the matrix $X$ defined by the columns of $X$ indexed by the set $S\subseteq \{1,\dots,n\}$.
Let $X_S^\inv$ be the Moore-Penrose pseudo-inverse of the matrix $X_S$.
It is well-known (and easy to verify) that the minimizer of $\|XW - Y\|_F^2$ subject to $W$ whose non zero rows are indexed by $S$ is equal to $\|Y-X_SX_S^\inv Y\|^2_F$.
SMLR can therefore be reformulated as 
$$\min_{S\subseteq [n]} \{f(S) = \|Y-X_SX_S^\inv Y\|^2_F: |S|\le k\} \ .$$
We can consequently apply our methodology from Section \ref{main} to SMLR if we show that $f(S)$ is $\alpha$-weakly-supermodular.

\begin{lemma}
For $X \in \R^{m \times n}$ and $Y \in \R^{m \times \ell}$ the 
SMLR minimization function $f(S) = \|Y-X_SX_S^\inv Y\|^2_F$ is $\alpha$-weakly-supermodular with 
$\alpha = \max_{S'} \|X_{S'}^\inv \|^2_2$.
\end{lemma}
\begin{proof}
We first estimate $f(S)-f(S\cup T)$. Denote by $Z_{T\setminus S}$ the matrix whose columns are those of $X_{T\setminus S}$ projected away from the span $X_S$ and normalized.
More formally, $\zeta_i = \|(I - X_{S} X_{S}^\inv)x_i\|$ and $z_i = (I - X_{S} X_{S}^\inv)x_i/\zeta_i$ for all $i\in T\setminus S$. 
Note that the column span of $Z_{T \setminus S}$ is orthogonal to that of $X_S$ and that together they are equal to the column span of $X_{T \cup S}$.
Using the Pythagorean theorem we obtain $f(S) = \|Y\|_F^2 - \|X_{S}X_S^\inv Y\|_F^2$ and $f(S \cup T) = \|Y\|_F^2 - \|X_{S}X_S^\inv Y\|_F^2 - \|Z_{S\setminus T}Z_{S\setminus T}^\inv Y\|_F^2$.
Substituting $T= \{i\}$ also gives $f(S)-f(S\cup \{i\})= \|z_{i} z_{i}^\T Y  \|_F^2$.
\begin{align}\label{newAlphaNatarajan}
f(S)-f(S\cup T) &= \|Z_{T\setminus S}Z_{T\setminus S}^\inv  Y  \|_F^2 \\
& = \|(Z^{\T}_{T\setminus S})^\inv \cdot  Z^\T_{T\setminus S} Y  \|_F^2  & \mbox{ by Singular Value Decomposition} \\
&\le \|(Z^{\T}_{T\setminus S})^\inv  \|_2^2 \cdot \|Z^\T_{T\setminus S} Y  \|_F^2  \\
&= \|Z_{T\setminus S}^\inv  \|_2^2 \cdot \sum_{i \in T\setminus S} \|z_i^\T Y  \|_2^2   \\
&\le \|X_{T\cup S}^\inv  \|^2_2 \cdot |T\setminus S| \max_{i \in T\setminus S} \|z^\T_i Y  \|_2^2 & \mbox{ see below}\label{switchZX}  \\ 
&\le \alpha \cdot |T\setminus S|  \left[ f(S)-f(S\cup \{i\}) \right]   & 
\end{align}
For Equation~\eqref{switchZX} we use a non trivial transition, $\|Z_{T\setminus S}^\inv  \|_2 \le \|X_{T\cup S}^\inv  \|_2$.
%
%
By the definition of $ Z_{T\setminus S}$ we can write for $i \in T\setminus S$ that $z_i=(x_i-\sum_{j\in S}\alpha_{ij} x_j)/\zeta_i$ and $\zeta_i = \|(I - X_{S} X_{S}^\inv )x_i\|$.
For any vector $w \in \R^{|T\setminus S|}$
$$  Z_{T\setminus S}w = \sum_{i \in T\setminus S} x_i w_i/\zeta_i  + \sum_{j \in S}  x_j \sum_{i \in T \setminus S} w_i\alpha_{ij}/\zeta_i = X_{T \cup S}w'$$
where $w'_i = w_i/\zeta_i$ for $i \in T\setminus S$ and $w'_j =\sum_{i \in T \setminus S} w_i\alpha_{ij}/\zeta_i$ for $j \in S$.  
Since, $\zeta_i = \|(I - X_{S} X_{S}^\inv )x_i\| \le \|x_i\| = 1$ we have $\|w'\| \ge \|w\|$. 
Finally, consider $w$ such that $\|w\|=1$ and $\|Z_{T\setminus S}w\| = \|Z_{T\setminus S}^\inv \|^{-1}$. 
This is the right singular vector corresponding to the smallest singular value of $Z_{T\setminus S}$. We obtain
\[
\|Z_{T\setminus S}^\inv \|^{-1} = \|Z_{T\setminus S}w\| = \|X_{T \cup S}w'\| \ge \|X_{T \cup S}^\inv \|^{-1}\|w'\| \ge \|X_{T \cup S}^\inv \|^{-1} \ .
\]
Which completes the proof.
\end{proof}
\begin{lemma}\label{slow_kmeans}
Let $f(S^*)$ be the optimal solution to the Sparse Multiple Linear Regression problem.
One can find in time $O(\alpha k\log (\|Y\|^2_F/\eps)\cdot nT_f)$ a set $S\subseteq  [n]$ of size $|S| =  \lceil \alpha k \log (\|Y\|^2_F/\eps) \rceil$ such that $f(S) \le (1+\eps)f(S^*)$ where $T_f$ is the time needed to compute $f(S)$ once.
\end{lemma}

\section{Sparse Regression}
The problem of Sparse Regression defined in \cite{Natarajan1995} is an instance of SMLR where the number of columns in $Y$ is $\ell=1$.
Since both $Y$ and $W$ are vectors we reduce the more familiar form of this problem; minimize $\|Xw - y\|_2^2$ subject to $\|w\|_{0} \le k$.

\cite{Natarajan1995} analyzed the greedy algorithm for the sparse regression problem.
He sets a desired threshold error $E$ and defined $k$ to be the minimum cardinality of a solution $S^*$ that achieves $f(S^*) \le E' = E/4$.  
He showed that the greedy algorithm finds a solution $S$ such that $f(S)\le E$ such that 
$$|S| \le \left\lceil 9 k\cdot \|X^\inv \|_2^2 \ln \frac{\|y\|^2_2}{E}\right\rceil$$
In his work \cite{Natarajan1995} implicitly assumes the over constrained setting 
where the number of columns $m$ in $X$ is smaller than their dimension $n$ and that $X$ is full rank. 
In this setting $\alpha = \max_{S'} \|X_{S'}^\inv \|  = \|X^\inv \|$ by Cauchy's interlacing theorem.

Here, we apply Theorem \ref{modification} with initial solution $S_0=\emptyset$ (which gives $f(S_0)=\|y\|_2^2$) and $E'=E/4$.
It immediately yields that the greedy algorithm finds a solution of value $f(S)\le E$ such that 
$$|S|\le   \left\lceil k\cdot \|X^\inv \|_2^2 \ln \frac{\|y\|_2^2}{E - E/4}\right\rceil \le   \left\lceil k\cdot \|X^\inv \|_2^2 \left( \ln \frac{\|y\|_2^2}{E}  + \ln \frac{4}{3} \right) \right\rceil $$
This improves the result of \cite{Natarajan1995} in three ways 
1) the approximation factor is smaller by a constant factor 
2) its proof is more streamlined and 
3) it is extended to viability of the greedy algorithm to the under constrained case where the result of \cite{Natarajan1995} does not hold. 
Specifically, where his implicit assumption that $\max_{S'} \|X_{S'}^\inv \|  = \|X^\inv \|$ no longer holds.

\section{Column Subset Selection Problem}
Given a matrix $X$, Column Subset Selection (CSS) is concerned with finding a small set of columns whose span captures as much of the Frobenius norm of $X$.
It was throughly investigated in the context of numerical linear algebra \cite{SSP1, SSP2, SSP3}.
In other words, find a subset $S \in [n]$, $|S|\le k$ of matrix columns the minimize $f(S) = \|X- X_SX_S^\inv X\|_F^2$.
This formulation makes it clear that this is a special case of SMLR where $Y = X$.

\cite{SviridenkoW13} investigated notion of a curvature $c\in [0,1]$ for a nonincreasing set functions. They define it as follows:
%
\begin{equation}
\label{eq:gen-curvature}
c = 1 - \min_{j \in [n]} \min_{S,T \subseteq [n] \setminus \{j\}}\frac{f(S)-f(S\cup\{j\})}{f(T)-f(T\cup\{j\})}.
\end{equation}
They show that there exists a greedy type algorithm that finds a solution of value at most $1/(1-c)$ times the optimal 
value of the minimization problem for any objective set function with curvature $c$ (Corollary 8.5 in \cite{SviridenkoW13}).
\begin{lemma}[Lemma 9.1 from \cite{SviridenkoW13}]\label{knownFacts}
Let $f(S)$ be the objective function for the Column Subset Selection Problem corresponding to the matrix $X$. 
The curvature $c$ of $f(S)$ is such that $\frac{1}{1-c}\le \kappa^2(X)$ 
where $\kappa(X)$ is the condition number of $X$.
\end{lemma}
Note that for any matrix $X$with full column rank if $\tilde{X}$ is the matrix with normalized columns then $\|\tilde{X}^+\|\le \kappa(X)$.
We can find our initial solution $S_0$ by one of the three known methods:
\begin{enumerate}
\item an approximation  algorithm from \cite{SviridenkoW13} finds a solution $S_0$ such that $|S_0|=k$ and performance guarantee $\rho=\kappa^2(X)$;
\item an approximation  algorithm from \cite{apx1,apx2} with $|S_0|=k$ and  $\rho=k+1$;
\item an approximation  algorithm from \cite{CB} with $|S_0|=2k$ and  $\rho=2$;
\end{enumerate}
\begin{lemma}
For the columns subset selection problem for a column normalized matrix $X$ and $\alpha = \max_{S'} \|X_{S'}^\inv \|^2_2$ one can fine a set $S$ of value $f(S)\le (1+\delta)f(S^*)$ such that 
$$|S| = O\left(\alpha k \left( \ln \frac{\rho }{\delta} \right) \right).$$
\end{lemma}
\begin{proof}
Combining one of the above results with the algorithm from Section \ref{main} completes the proof.
\end{proof}
\section{Acknowledgments}
We would like to thanks Sergei Vassilvitskii and Dan Feldman for their guidance and for Petros Drineas and for pointing out the vulnerability of Natarajan's proof.

\bibliographystyle{unsrt}

\end{document}